\documentclass[a4paper,USenglish]{lipics-v2021}

\hideLIPIcs  
\nolinenumbers 

\title{Database Theory in Action: From Inexpressibility to Efficiency in GQL's Order-Constrained Paths}
\titlerunning{Inexpressibility to Efficiency in GQL's Order-Constrained Paths}

\author{Hadar Rotschield}{School of Computer Science and Engineering, The Hebrew University of Jerusalem, Israel}{hadar.rotschield@mail.huji.ac.il}{https://orcid.org/0009-0000-2230-1734}{}
\author{Liat Peterfreund}{School of Computer Science and Engineering, The Hebrew University of Jerusalem, Israel}{liat.peterfreund@mail.huji.ac.il}{https://orcid.org/0000-0002-4788-0944}{}
\authorrunning{H. Rotschield and L. Peterfreund}

\Copyright{H. Rotschield and L. Peterfreund} 



\usepackage[T1]{fontenc}
\usepackage{lmodern}
\usepackage[utf8]{inputenc}
\usepackage{subcaption}
\usepackage{graphicx} 
\usepackage[normalem]{ulem}
\usepackage{subcaption}
\usepackage{algorithm}
\usepackage{algpseudocode}
\usepackage{booktabs}
\usepackage{url}
\usepackage{xurl}
\algrenewcommand\algorithmicrequire{\textbf{Input:}}
\algrenewcommand\algorithmicensure{\textbf{Output:}}

\newenvironment{reptheorem}[1]{%
  \begingroup
  \def\thetheorem{\ref{#1}}%
  \begin{theorem}%
}{%
  \end{theorem}%
  \endgroup
}

\newcommand{\src}{\mathsf{src}}
\newcommand{\tgt}{\mathsf{tgt}}

\newcommand{\OMIT}[1]{}

\newcommand{\df}{:=}
\newcommand{\query}{\mathcal Q}

\usepackage{algorithm}
\usepackage{algpseudocode}
\usepackage{multicol}
\usepackage{amsmath}
\newenvironment{repproposition}[1]{%
  \begingroup
  \def\thetheorem{\ref{#1}}%
  \begin{proposition}%
}{%
  \end{proposition}%
  \endgroup
}

\supplement{}
\supplementdetails[subcategory={Source Code}, cite={}, swhid={}]{Software}{https://github.com/hadarrot/Order-Constrained-Paths-in-Leveled-Graphs}
\funding{Both authors were supported by ISF grant 2355/24.} 

\begin{document}
\maketitle

\begin{abstract}
Pattern matching of core GQL, the new ISO standard for querying property graphs, cannot check whether edge values are increasing along a path, as established in recent work.
We present a constructive translation that overcomes this limitation by compiling the increasing-edges condition into the input graph.
Remarkably, the benefit of this construction goes beyond restoring expressiveness. In our proof-of-concept implementation in Neo4j's Cypher, where such path constraints are expressible but costly, our compiled version runs faster and avoids timeouts. 
This illustrates how a theoretically motivated translation can not only close an expressiveness gap but also bring practical performance gains.
\end{abstract}
\section{Introduction}

GQL is the new ISO standard for querying property graphs~\cite{GQLStandards}. Property graphs are a rich and flexible data model in which nodes and edges can carry labels (often representing their types) and properties (in a key-value format), which makes them suitable for various domains such as finance, social networks, and knowledge graphs~{\cite{srivastava2023fraud, ArenasGutierrezSequeda2021, ZhuNieLiuZhangWen2009}.}

The GQL standard is lengthy and complex, first formalized in~\cite{pods23,icdt23}, later distilled into its core~\cite{vldb25}, enabling precise expressiveness analysis, {and subsequently positioned on a theory/practice spectrum as articulated in~\cite{libkin-martens-murlak-peterfreund-vrgoc-pods25}.}
{GQL comprises a pattern-matching layer,
used in MATCH to produce bindings of nodes or edges, and
a relational layer that post-processes those bindings using relational algebra.}
One of the key findings in~\cite{vldb25} is that queries identifying paths with increasing edge values are not expressible in core GQL. Such queries are common in practice, leading participants in the ISO GQL design to consider possible language extensions to support them~\cite{fred,tobias}. For example, in a financial transaction graph, accounts are nodes with properties such as balance, and transfers are edges with properties such as amount and timestamp. While pattern matching allows to query for chains of transfers where account balances increase along the path, expressing that timestamps increase along a path is impossible. In other words, increasing value conditions can be checked on nodes, but not on edges.

In this paper, we present a constructive translation that overcomes the expressiveness gap by compiling an input graph and order condition into a leveled graph which, together with a  reachability query,  captures the intended semantics. Our contribution is twofold: first, we provide a principled method that makes inexpressible queries formally definable; second, we show that the same approach yields practical benefits. While such queries can be expressed in the full versions of GQL and in Cypher, the formulations are highly convoluted and require reasoning over exponentially many paths, leading to impractical runtimes as demonstrated in the experiments in~\cite{vldb25}. In contrast, our technique achieves efficient execution in practice, as confirmed by a proof-of-concept implementation. This demonstrates how a theoretically motivated translation can both extend the expressive power of a standard language and provide concrete performance improvements.

\section{Compiling Ordered-Path Constraints into a Leveled Graph}
In this section we show how to compile the increasing values along edges condition into the input graph.
We start by defining the data model. We assume pairwise disjoint sets \textit{Nodes}, \textit{Edges}, \textit{Labels}, \textit{Properties}, and \textit{Values}, and define 
a \emph{labeled property graph} (see, e.g.,~\cite{libkin-martens-murlak-peterfreund-vrgoc-pods25}) as a tuple
\(
G=(N,E,\mathrm{src},\mathrm{tgt},\lambda,\rho)
\)
where \(N\subseteq\textit{Nodes}\) and \(E\subseteq\textit{Edges}\) are finite,
\(\mathrm{src},\mathrm{tgt}:E\to N\) are total functions,
\(\lambda:E\to\textit{Labels}\) assigns an edge label to each edge, and
\(\rho:(N\cup E)\times\textit{Properties}\rightharpoonup \textit{Values}\) is a partial function that stores node/edge properties.

\smallskip
\textbf{Problem Definition.}
To formulate the problem, we assume each edge is associated with a designated property $\textsf{val}$ whose value is in $\mathbb R$. Formally, we assume \(\rho(e,\textsf{val})\) is defined for  every \(e \in E\) and that  \( \rho(e,\textsf{val}) \in \mathbb{R}\).
For readability,  we denote $ u \xrightarrow{j} v$ if there is a directed edge $e$ from $u$ to $v$ with  \(j\df \rho(e,\textsf{val})\).
It was shown in~\cite{vldb25} that the following problem is not expressible in pattern matching of core GQL:

\begin{center}
\fbox{
\begin{minipage}{0.96\linewidth}
\label{def:srtictly-inc-path}
\textbf{Strictly Increasing-Path Existence.}
For every labeled property graph $G$, and nodes $s,t\in N$, does there exist a non-empty directed path
\(
s=u_0 \xrightarrow{j_1} u_1 \xrightarrow{j_2} \cdots \xrightarrow{j_m} u_m=t \quad (m\ge 1)
\)
whose edge-values are strictly increasing  (i.e.,
\(
j_1 \;<\; j_2 \;<\; \cdots \;<\; j_m
\))?
\end{minipage}}
\end{center}

The underlying reason that this query is not expressible is that GQL’s concatenation, by design, verifies only node equivalence.

\subsection*{Our Construction}
We show that given a graph, we can compile the increasing value condition into it. 
The idea is to encode, in the node identifier of the compiled graph, the value of an incoming edge, hence reducing  the increasing value conditions to  ordinary reachability.

\begin{definition}[Leveled graph]\label{def:leveled}
Let $G$ be a labeled property graph.
For each node $v$ of $G$, let
$
L(v) = \{\, j \mid u \xrightarrow{j} v \text{ for some } u \,\} \cup \{\bot\},
$
that is, $L(v)$ is the set of values of the incoming edges to $v$, together with a special value $\bot$ used for paths that start at $v$. We assume $\bot < j$ for all values $j$.
The leveled graph $G^{\mathsf{lev}}$ of $G$ consists of nodes $(v,j)$ for every edge $u \xrightarrow{j} v$, and, for every $v\in N$, the node $(v,\bot)$, and edges
$
(u,\ell) \to (v,j)
$ for every  edge $u \xrightarrow{j} v$, and $\ell \in L(u)$ with $\ell < j$.
\end{definition}

\begin{figure}[t]
	\centering
	\makebox[\linewidth][c]{
		\begin{subfigure}{.3\linewidth}
			\centering
			\includegraphics[width=\linewidth]{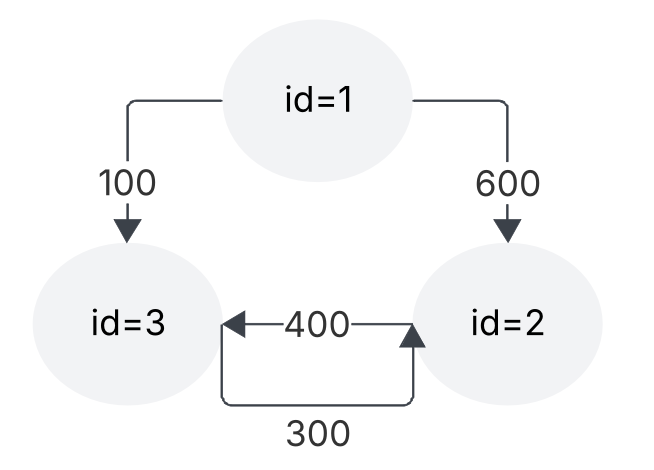}
			\caption{\label{fig:trans} Input property graph $G$.}
			\label{fig:ex-base}
		\end{subfigure}\hspace{1.2em}
		\begin{subfigure}{.43\linewidth}
			\centering
			\includegraphics[width=\linewidth]{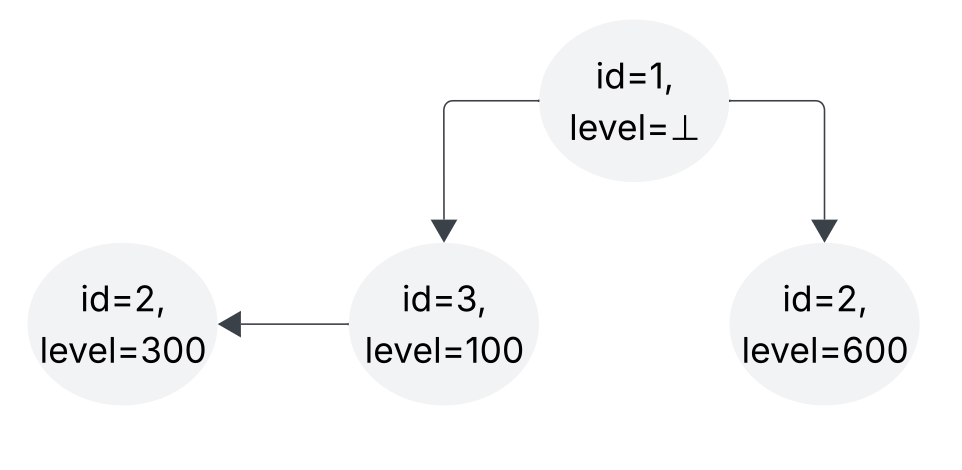}
			\caption{Leveled graph $G^{\mathsf{lev}}$.}
			\label{fig:ex-leveled}
		\end{subfigure}
	}
	\caption{\label{fig:ex-lev}Leveled-graph example (with amounts). In (a) the increasing path
		$1\!\to\!3\!\to\!2$ exists ($100<300$), while $1\!\to\!2$ with $600$ cannot be
		extended via $2\!\to\!3$ with $400$. In (b) nodes are copied by their last-seen
		from node $1$ to some node $2$ in $G^{\mathsf{lev}}$.}
	\label{fig:ex-leveled-example}
\end{figure}

To illustrate the definition, consider the following example. 

\begin{example}[Increasing amounts]
\label{ex:increasing}
The graph in Figure~\ref{fig:trans}  describes transfers between accounts where $j$ denotes the amount of the transfer for each edge $u \xrightarrow{j} v$.
In 
$G^{\mathsf{lev}}$ in Figure~\ref{fig:ex-leveled}, each account $v$ is replicated as $(v,\ell)$ for every $\ell\in L(v)$.
There is an edge $(u,\ell)\to(v,j)$ iff there is a transfer $u \xrightarrow{j} v$ in $G$ and { there exists  $\ell \in L(u)$ for which $\ell<j$}.
Thus, an increasing-transfer path from $s$ to $t$ in $G$ corresponds to a reachability problem in $G^{\mathsf{lev}}$. 
\end{example} 

With this intuition, we can move to showing the correctness of our construction.

\newcommand{\thmLevelLiftcorrect}{
Fix a property graph $G$ and nodes $s,t\in N$. There exists a non-empty strictly increasing path
$s \rightsquigarrow t$ in $G$ iff there exists $\ell\in L(t)$ such that
$(s,\perp)\rightsquigarrow (t,\ell)$ in $G^{\mathsf{lev}}$.
}

\begin{theorem}
\label{thm:correct}
  \thmLevelLiftcorrect
\end{theorem}

\begin{proof}[Proof Sketch]
($\Rightarrow$) Let $s\leadsto t$ in $G$ be $u_0\!\xrightarrow{j_1}\!u_1\!\xrightarrow{j_2}\!\cdots\!\xrightarrow{j_m}\!u_m$ with $j_1<\cdots<j_m$.
As $\bot<j_1$
and $j_{i-1}<j_i$, Definition~\ref{def:leveled} creates the path
$(u_0,\bot)\to(u_1,j_1)\to\cdots\to(u_m,j_m)$ in $G^{\mathsf{lev}}$.
 

($\Leftarrow$) Any path $(s,\bot)\to(u_1,j_1)\to\cdots\to(t,\ell)$ in $G^{\mathsf{lev}}$
projects to $s\!\xrightarrow{j_1}\!\cdots\!\xrightarrow{j_m}\!t$ in $G$.
Definition~\ref{def:leveled} ensures the edge conditions
enforce $j_{i-1}<j_i$.

\end{proof}
This immediately yields the desired reduction: checking for a strictly increasing path in $G$ reduces to a reachability query in $G^{\mathsf{lev}}$. 

\smallskip
\textbf{Time Complexity.}
Using the bounds below, the construction of the leveled graph can be done
in time polynomial
in its input size.
\newcommand{\propTimeComplexity}{
One can construct \(G^{\mathsf{lev}}\) from \(G\) by: (i) computing $L(v)$ for each node $v\in N$
	(ii) creating a node \((v,\ell)\) for each \(v\in N\) and \(\ell\in L(v)\),
	and (iii) for each edge \(u\xrightarrow{j} v\) in \(G\) and each \(\ell\in L(u)\) with \(\ell<j\),
	adding the edge \((u,\ell)\to(v,j)\).
    This can be done in 
$O(|N|+|E|^2)$ time.
}

\begin{proposition}
\label{prop:time-complexity}
  \propTimeComplexity
\end{proposition}

\begin{proof}[Proof Sketch]
Recall that $L(v)=\{\,j: \exists u,\,u\xrightarrow{j}v\in E\,\}\cup\{\bot\}$ hence we can compute $L(v)$ for all $v\in N$ in $O(|E|)$. 
{For $N'$ and $E'$ the nodes and edges of \(G^{\mathsf{lev}}\), respectively, we have}

\[
\begin{aligned}
|N'|
&=\sum_{v\in N}\! |L(v)|
\;{\le}\; \sum_{v\in N}\! (1+\mathrm{indegree}(v))
\;{=}\; |N| + |E|.
\\
|E'|
&= \sum_{(u,j,v)\in E} \bigl|\{\;\ell\in L(u)\mid \ell<j\;\}\bigr|
\;\le\; \sum_{(u,j,v)\in E} |L(u)|
\;\le\; |E| \cdot \max_{u\in N} |L(u)|
\;\le\; |E|^2.
\end{aligned}
\]
{\noindent Hence, to compute $G^{\mathsf{lev}}$ we need $O(|E|+|N'|+|E'|)=O(|N|+|E|^2)$ time.}
\end{proof}

The increasing-path condition is expressible in the full language with
set operations (in particular, difference)~\cite{vldb25}.
Here we deliberately work
at the pattern-matching level because composing
the required set operations leads to blow-ups, whereas compiling the order constraint into the leveled graph lets
us answer it via plain reachability.

\section{Experiments}
\label{sec:experiments}

\newlength{\panelheight}
\setlength{\panelheight}{0.35\textheight} 

\begin{figure}[t]
\centering
\begin{minipage}[t][\panelheight][b]{0.45\linewidth} 
  \centering
  \resizebox{\linewidth}{!}{
  \begin{tabular}{@{}rcccc@{}}
  \toprule
  \shortstack{$|E|$} &
  \shortstack{Leveled\\build (ms)} &
  \shortstack{Leveled\\query (ms)} &
  \shortstack{Baseline\\query (ms)} &
  \shortstack{Speedup\\($\times$)} \\
  \midrule
   20  &  57 & 10 &    9 & 0.13 \\
   40  &  17 &  6 &    6 & 0.26 \\
   60  &  12 &  6 &    7 & 0.39 \\
   80  &  76 &  6 &    8 & 0.10 \\
  100  &  16 &  8 &   12 & 0.50 \\
  120  &  20 &  7 &   10 & 0.37 \\
  140  &  15 &  7 &  363 & \textbf{16.50} \\
  160  &  15 &  7 & 1191 & \textbf{54.14} \\
  180  &  17 &  7 & \textsc{timeout} & --- \\
  200  &  23 &  7 & \textsc{timeout} & --- \\
  220  &  29 & 11 & \textsc{timeout} & --- \\
  240  &  21 &  7 & \textsc{timeout} & --- \\
  260  &  80 &  7 & \textsc{timeout} & --- \\
  280  &  28 &  8 & \textsc{timeout} & --- \\
  300  &  16 &  8 & \textsc{timeout} & --- \\
  \bottomrule
\end{tabular}

  }
  \captionof{table}{Baseline vs.\ leveled at $|N|{=}100$.
  Speedup is (baseline avg query)/({leveled build}+{leveled avg query}).}
  \label{tab:results-leveled}
\end{minipage}\hfill
\begin{minipage}[t][\panelheight][b]
{0.5\linewidth} 
  \centering
  \includegraphics[width=\linewidth]{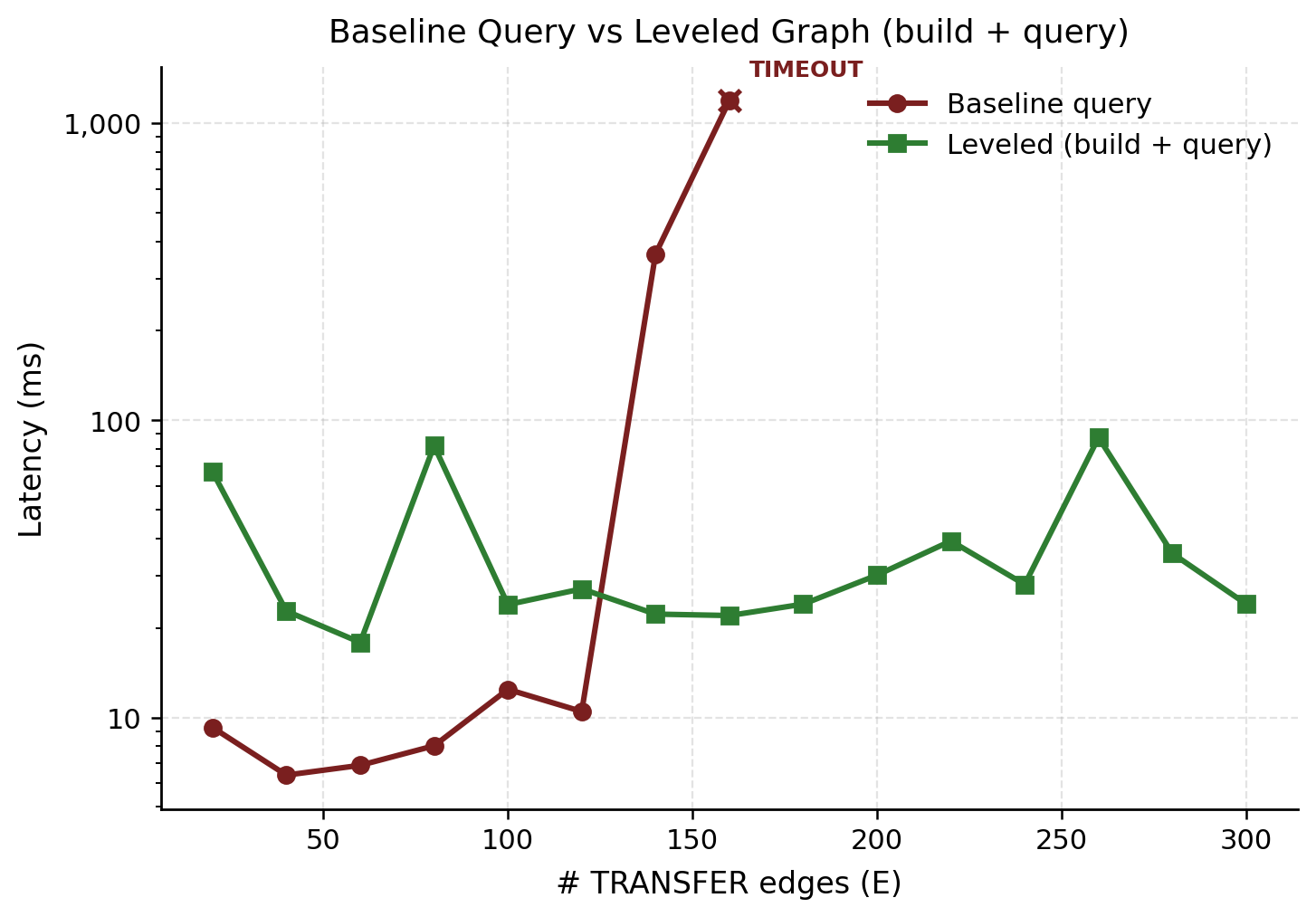}
  \caption{Per-run average latency (log-$y$). Baseline rises and times out for $|E|\!\ge\!180$ (10s limit), leveled query and build remains $\le\!100$\,ms across these sizes.}
\label{fig:leveled-lat}
\end{minipage}
\end{figure}

We evaluate the leveled graph compilation for strictly increasing edge values by comparing:
(i) a {baseline} query on the original graph that enforces $j_{i-1}<j_i$, and
(ii) a {leveled} graph build along with a  reachability query on the compiled leveled graph $G^{\mathrm{lev}}$.

\smallskip
\textbf{Setup.}
We use a local Neo4j instance accessed via the {official Neo4j \texttt{python} driver}.\footnote{\scriptsize Experiments were run locally on macOS {15.6} (build {24G84}) on a MacBook Air ({Mac16,13}) with an {Apple M4} CPU ({10} cores: 4 performance + 6 efficiency) and {16\,GB} RAM.
The database was {Neo4j 2025.07.1} (local install, default configuration), via the official {Neo4j Python driver}.
Client stack: {Python 3.13.5} and {neo4j driver 5.28.2}.} The script rebuilds the base graph for each edge count, times the base build, compiles and times $G^{\mathrm{lev}}$, and then runs repeated baseline/leveled queries with client-side per-run timeouts enforced by a worker process. The workload fixes $|N|{=}100$ accounts and varies $|E|$ from $20$ to $300$ in {jumps of 20}, each edge has an integer \texttt{amount} in $[1,1000]$. We run 10 times per $|E|$ with timeout=10 sec. Averages are over successful runs only.

\smallskip
\textbf{Results.}
Leveled reachability is nearly flat across densities ($\sim$6–11\,ms per run), whereas the baseline grows sharply and then times out. The leveled approach is slower for small graphs ($|E|\le 120$), but becomes decisively faster as density grows: at $|E|{=}140$ we observe a \textbf{16.50}$\times$ speedup, rising to \textbf{54.14}$\times$ at $|E|{=}160$. For $|E|\ge 180$ all baseline runs hit the 10s per-run timeout, while leveled runs complete without timeouts. Note that results may vary with random graph instances, local resource contention and driver overhead. 
\section{Conclusions}
We revisited the expressiveness gap in core GQL and showed that compiling the order constraint into the graph 
reduces the inexpressible query 
to a standard reachability one while pointing to potential runtime benefits in our prototype implementation. 
Looking ahead, we propose to extend our method 
to address a broader class of queries (such as those including alternating order constraints) and to conduct a more systematic experimental study across different implementations, languages, datasets, and scales.

\newpage
\bibliographystyle{plainurl}
\bibliography{refs}

\newpage
\appendix 
\section*{Appendix}

\begin{reptheorem}{thm:correct}
  \thmLevelLiftcorrect
\end{reptheorem}

\begin{proof}\label{proof:theorem-correct}
We recall the construction of $G^{\mathsf{lev}}$ from Definition~\ref{def:leveled}.
For each $v\in N$ let
\[
L(v)\ :=\ \{\, j\in D \mid \exists u\in N:\ u \xrightarrow{j} v \in E \,\}\ \cup\ \{\bot\}.
\]
The leveled graph is \(G^{\mathsf{lev}}=(N',E',\src',\tgt',\lambda',\rho')\) with
\(
N' \ =\ \{\, (v,\ell)\mid v\in N,\ \ell\in L(v)\,\},
\)
and for every \(e\in E\) with \(j:=\rho(e,\textsf{val})\) and every \(\ell\in L(\src(e))\) such that \(\ell<j\),
we create a fresh \(e_\ell\in E'\) and set
\(
\src'(e_\ell)=(\src(e),\ell),\qquad \tgt'(e_\ell)=(\tgt(e),j),
\)
Increasing paths in $G$ is a (non-empty) path in $G$ is a sequence
\[
\pi \ :=\ u_0 \xrightarrow{j_1} u_1 \xrightarrow{j_2} \cdots \xrightarrow{j_m} u_m
\quad (m\ge 1),
\]
and it is {strictly increasing} if $j_1<j_2<\cdots<j_m$.

When we say projection from $G^{\mathsf{lev}}$ to $G$, we mean that each Leveled edge $((u,\ell),(v,j))\in E'$ is by construction associated with the unique
base edge $u\xrightarrow{j} v\in E$ and satisfies the side condition $\ell<j$.

Given a (non-empty) Leveled path
\[
\hat\pi \ :=\ (u_0,\ell_0)\to(u_1,\ell_1)\to\cdots\to(u_m,\ell_m)
\]
in $G^{\mathsf{lev}}$, its {projection} to $G$ is the sequence
\[
\pi \ :=\ u_0 \xrightarrow{j_1} u_1 \xrightarrow{j_2} \cdots \xrightarrow{j_m} u_m,
\]
where each Leveled edge $((u_{i-1},\ell_{i-1}),(u_i,\ell_i))$ corresponds to the base
edge $u_{i-1}\xrightarrow{j_i}u_i$.

\medskip
\noindent(\(\Rightarrow\)) Suppose there exists a strictly increasing path
$\pi = u_0 \xrightarrow{j_1} u_1 \xrightarrow{j_2} \cdots \xrightarrow{j_m} u_m$
in $G$ from $s=u_0$ to $t=u_m$ with $j_1<\cdots<j_m$.
Define $\ell_0:=\bot$ and, for $i\ge 1$, set $\ell_i:=j_i$.
We show by induction on $i=1,\ldots,m$ that $((u_{i-1},\ell_{i-1}),(u_i,\ell_i))\in E'$.

{Base $i=1$.} Since $\bot<j_1=\ell_1$ and $u_0\xrightarrow{j_1}u_1\in E$, the edge
$((u_0,\bot),(u_1,j_1))$ is in $E'$.

{Step $i\to i{+}1$.} We have $u_i\xrightarrow{j_{i+1}}u_{i+1}\in E$ and
$\ell_i=j_i<j_{i+1}=\ell_{i+1}$.
Because $j_i$ is the value on an edge incoming to $u_i$, we have $\ell_i\in L(u_i)$,
so $((u_i,\ell_i),(u_{i+1},\ell_{i+1}))\in E'$.

Thus $\hat\pi:(s,\bot)\to(u_1,j_1)\to\cdots\to(t,j_m)$ is a Leveled path in $G^{\mathsf{lev}}$.
Since $j_m\in L(t)$, its endpoint has the form $(t,\ell)$ with $\ell\in L(t)$.

\medskip
\noindent(\(\Leftarrow\)) Suppose there exists a Leveled path
\[
\hat\pi \ :=\ (s,\bot)=(u_0,\ell_0)\to(u_1,\ell_1)\to\cdots\to(u_m,\ell_m)=(t,\ell)
\]
in $G^{\mathsf{lev}}$.
By construction of $E'$, each Leveled edge $((u_{i-1},\ell_{i-1}),(u_i,\ell_i))$
comes from a base edge $u_{i-1}\xrightarrow{j_i}u_i\in E$ with the side condition
$\ell_{i-1}<j_i$ and {target level} $\ell_i=j_i$.
Hence the projection
\[
\pi \ :=\ u_0 \xrightarrow{j_1} u_1 \xrightarrow{j_2} \cdots \xrightarrow{j_m} u_m
\]
is a path in $G$ from $s$ to $t$, and it is strictly increasing since
$\ell_{i-1}<j_i$ and $\ell_{i-1}=j_{i-1}$ imply $j_{i-1}<j_i$ for all $i$.

\medskip
Combining the two directions yields the equivalence stated in Theorem~\ref{thm:correct}.
\end{proof}

\begin{repproposition}{prop:time-complexity}
  \propTimeComplexity
\end{repproposition}

\begin{proof}\label{proof:prop-time-complexity}

The following Algorithm~\ref{alg:LevelLift} completes the proof.
\end{proof}

\begin{algorithm}[b!]
\fontsize{11}{13}\selectfont 
\caption{Constructing the leveled graph $G^{\mathsf{lev}}$.}
\label{alg:LevelLift}
\begin{algorithmic}[1]
\Require Graph $G$, each edge is written $u \xrightarrow{j} v$ with $j$ in a totally ordered domain $(D,<)$, fresh $\bot$ with $\bot<d$ for all $d\in D$.
\Ensure $G^{\mathsf{lev}}$ as in Def.~\ref{def:leveled}.
\Statex

\State \textbf{Levels}
\ForAll{$v\in N$}
  \State $L(v)\gets \{\bot\}$
\EndFor
\ForAll{$u \xrightarrow{j} v \in E$}
  \State $L(v)\gets L(v)\cup\{j\}$
\EndFor
\Statex

\State \textbf{Nodes}
\State $N'\gets \emptyset$
\ForAll{$v\in N$}
  \ForAll{$\ell\in L(v)$}
    \State $N'\gets N'\cup\{(v,\ell)\}$
  \EndFor
\EndFor
\Statex

\State \textbf{Edges}
\State $E'\gets \emptyset$
\ForAll{$u \xrightarrow{j} v \in E$}
  \ForAll{$\ell\in L(u)$ \textbf{with} $\ell<j$}
    \State $E'\gets E'\cup\{(u,\ell)\to(v,j)\}$
  \EndFor
\EndFor

\State \Return $G^{\mathsf{lev}}$
\end{algorithmic}
\end{algorithm}

\end{document}